\documentclass{article}
\usepackage{ijcai23}

\usepackage{times}
\usepackage{soul}
\usepackage{url}
\usepackage[hidelinks]{hyperref}
\usepackage[utf8]{inputenc}
\usepackage[small]{caption}
\usepackage{graphicx}
\usepackage{amsmath}
\usepackage{amsthm}
\usepackage{booktabs}
\usepackage[switch]{lineno}
\usepackage{amssymb}
\usepackage{enumerate}

\usepackage{float}
\usepackage{balance}
\usepackage{mathrsfs}
\usepackage{color}
\usepackage{subcaption}
\usepackage{autobreak}

\usepackage[linesnumbered,ruled,vlined]{algorithm2e}
\usepackage{algpseudocode}

\newtheorem{theorem}{Theorem}[section]
\newtheorem{definition}[theorem]{Definition}
\newtheorem{corollary}[theorem]{Corollary}
\newtheorem{lemma}[theorem]{Lemma}

\newtheorem{example}{Example}

\title{Truthful Auctions for Automated Bidding in Online Advertising}

\author{ Yidan Xing$^1$ \and
Zhilin Zhang$^2$\and
Zhenzhe Zheng$^{1,}$\thanks{Zhenzhe Zheng is the corresponding author.}\and
Chuan Yu$^2$ \and \\
Jian Xu$^2$ \and
Fan Wu$^1$ \And
Guihai Chen$^1$\\
\affiliations
$^1$Department of Computer Science and Engineering, Shanghai Jiao Tong University\\
$^2$Alibaba Group\\
\emails
$\{$katexing, zhengzhenzhe$\}$@sjtu.edu.cn,
$\{$zhangzhilin.pt, yuchuan.yc, xiyu.xj$\}$@alibaba-inc.com,\\
$\{$fwu, gchen$\}$@cs.sjtu.edu.cn
}

\begin{document}

\maketitle

\begin{abstract}
	Automated bidding, an emerging intelligent decision making paradigm powered by machine learning, has become popular in online advertising. Advertisers in automated bidding evaluate the cumulative utilities and have private financial constraints over multiple ad auctions in a long-term period. Based on these distinct features, we consider a new ad auction model for automated bidding: the values of advertisers are public while the financial constraints, such as budget and return on investment (ROI) rate, are private types. We derive the truthfulness conditions with respect to private constraints for this multi-dimensional setting, and demonstrate any feasible allocation rule could be equivalently reduced to a series of non-decreasing functions on budget. However, the resulted allocation mapped from these non-decreasing functions generally follows an irregular shape, making it difficult to obtain a closed-form expression for the auction objective. To overcome this design difficulty, we propose a family of truthful automated bidding auction with personalized rank scores, similar to the Generalized Second-Price (GSP) auction. The intuition behind our design is to leverage personalized rank scores as the criteria to allocate items, and compute a critical ROI to transform the constraints on budget to the same dimension as ROI. The experimental results demonstrate that the proposed auction mechanism outperforms the widely used ad auctions, such as first-price auction and second-price auction, in various automated bidding environments.
\end{abstract}

\section{Introduction}
    With the success of machine learning in online advertising \cite{zhang2014optimal,gharibshah2021user}, advertisers turned to adopting automated bidding (\emph{auto-bidding}) tools instead of bidding manually, bringing significant changes to the interaction between advertisers and online platforms \cite{web:google,web:facebook}.
    In auto-bidding services, advertisers submit their high-level optimization objectives and constraints to the platform, and then the bidding agents, powered by machine learning algorithms, make detailed bidding decisions in each of ad auctions on behalf of the advertisers. 
	With the help of automated bidding tools, advertisers can optimize their overall advertising objectives with respect to their financial constraints in a high-level way. 	
	\par Under automated bidding, we revisit a fundamental problem in auction theory: whether the conventional auction model, where advertisers have private values for items (\emph{i.e.}, ad impressions) and conduct corresponding strategic bidding for each single auction, is still appropriate for the new advertising paradigm.
	As the platform can access historical data about the interactions between advertisers and users, we can estimate the potential actions of users (such as clicks and conversions), which can be regarded as public values of items for advertisers.
	In auto-bidding, the private information from advertisers are actually their constraints for the whole advertising campaign. These features require a new ad auction model to incentivize advertisers to truthfully reveal the high-level private constraints given the values of items are public.  
	
	\par In this work, we consider a new automated bidding auction model, where advertisers submit budget as well as return on investment (ROI) requirement as their (private) constraints, and aim to maximize the cumulative values of winning impressions from multiple auctions during a certain period.
	We analyse the truthfulness conditions with respect to the private constraints of budget and ROI. 
	Remarkably, we show that any truthful auction mechanism for this multi-dimensional setting could be equivalently represented by a series of non-decreasing functions with budget as input. 
	When these non-decreasing functions are realized to derive the corresponding auction mechanism, the truthful conditions of budget and ROI introduce a new value grouping phenomenon: different budget-ROI types are grouped to share the same cumulative value, and the grouping pattern is determined by a threshold ROI function (transformed from the above non-decreasing function). As the threshold ROI functions are not constrained in monotonicity, the grouping shape of budget-ROI types is generally irregular, making it difficult to obtain the closed-form expression of grouping types and then the auction optimization objectives, such as revenue and social welfare.
	\par Facing these design difficulties, we propose a family of ad auctions with personalized rank scores to optimize various design objectives. Our auction adopts rank scores as the criteria to determine item allocations, which shares the similar ideas with Generalized Second-Price (GSP) auction~\cite{edelman2007internet}. 
	In guaranteeing truthfulness for the private constraints, we design \emph{critical ROI} to be the largest ROI that can win the most items without breaking the budget constraint.
	It equivalently transforms budget to the same dimension as ROI, thus allowing us to find a tight constraint limits the bidder from getting extra utilities and utilize this tight constraint to prevent misreporting. 
	We conduct extensive experiments to evaluate the performances of the proposed auction mechanism under various auto-bidding settings. The evaluation results demonstrate that the designed truthful auction can generally achieve more than $90$\% 
	performance (in terms of revenue and social welfare) of the optimal baselines without the consideration of truthfulness.
	The main contributions can be summarized as follows:
	
	\noindent $\bullet$ We consider the unique features within the interaction between advertisers and online platforms in the context of automated bidding. Based on these features, we formulate a new auto-bidding auction model, where value-maximizing bidders have high-level constraints as private information and the values of items are public. \\
	\noindent $\bullet$ We investigate truthfulness conditions of two-dimensional private constraints, budget and ROI, under public value setting. We provide full characterizations for the feasible space of allocation and payment for truthful auctions.\\
	\noindent $\bullet$ 
	Based on the derived truthfulness conditions, we design a family of truthful ad auction mechanisms for automated bidding. With the newly designed rank score functions, the proposed ad auction is simple and flexible to be adapted into various auto-bidding settings with different optimization goals such as revenue and social welfare. \\
	\noindent $\bullet$ We evaluate our proposed auction mechanisms with various experimental settings, and the empirical results validate the effectiveness of the proposed auctions regarding its performances in terms of social welfare and revenue.


	\section{Preliminaries} \label{sec:pre}
	In this section, we first motivate the considered auction design problem by the online advertising system with automated bidding services, and afterwards propose the formal auction model based on the features of automated bidding.
	\subsection{Online Advertising System}

	The working process of the online advertising auction system is illustrated in Figure \ref{fig:ad-system}. From an advertiser's perspective, it can be described as follows:
	\par 1) The advertiser sets the bidding configuration in auto-bidding interface: chooses optimization objectives (\emph{e.g.}, maximizing clicks or conversions) and sets cost constraints (\emph{e.g.}, budget per day, targeted return on investment (ROI) and maximum cost per click). The advertiser requires her realized ROI, defined as the ratio of her gained value and payment, to be higher than her targeted ROI.
	\par 2) Based on the advertiser's configuration, an auto-bidding agent represents the advertisers to make bid decisions in multiple auctions. When each user impression comes, the auto-bidding agents attend an ad auction to compete for the ad display opportunities. 
	\par 3) The auto-bidding agent adopts data-driven algorithms to predict the value (click through rate or conversion rate) of the incoming user impression, and bid for each impression while taking the cost constraints into consideration. 
	\par 4) During the multiple auctions, the advertiser can check the cumulative auction outcomes, including spent budget, received impressions, clicks, conversions and the average costs. The advertisers can further adjust their auto-bidding settings in the interface to achieve their own advertising objectives.
	\par From the above interaction between advertisers and the platform/auctioneer, we summarize three new features about the online advertising with auto-bidding services. First, advertisers only report \emph{high-level optimization objectives and constraints} in the bidding configurations for multiple auctions, but not the fine-grained bid for each auction. Second, advertisers only evaluate \emph{cumulative long-term performances and costs} of multiple auctions, instead of the outcome of each individual auction. Third, since the online platform can access all the data produced in advertising, it is reasonable to claim \emph{the monetary value} of an incoming impression to a specific advertiser can be calculated exactly. As the behavior patterns of advertisers and the online platform change distinctly in automated bidding, we need to investigate the mechanism design whose formats align with these new features.
	
	\begin{figure}[t]
		\centering
		\includegraphics[width=0.75\linewidth]{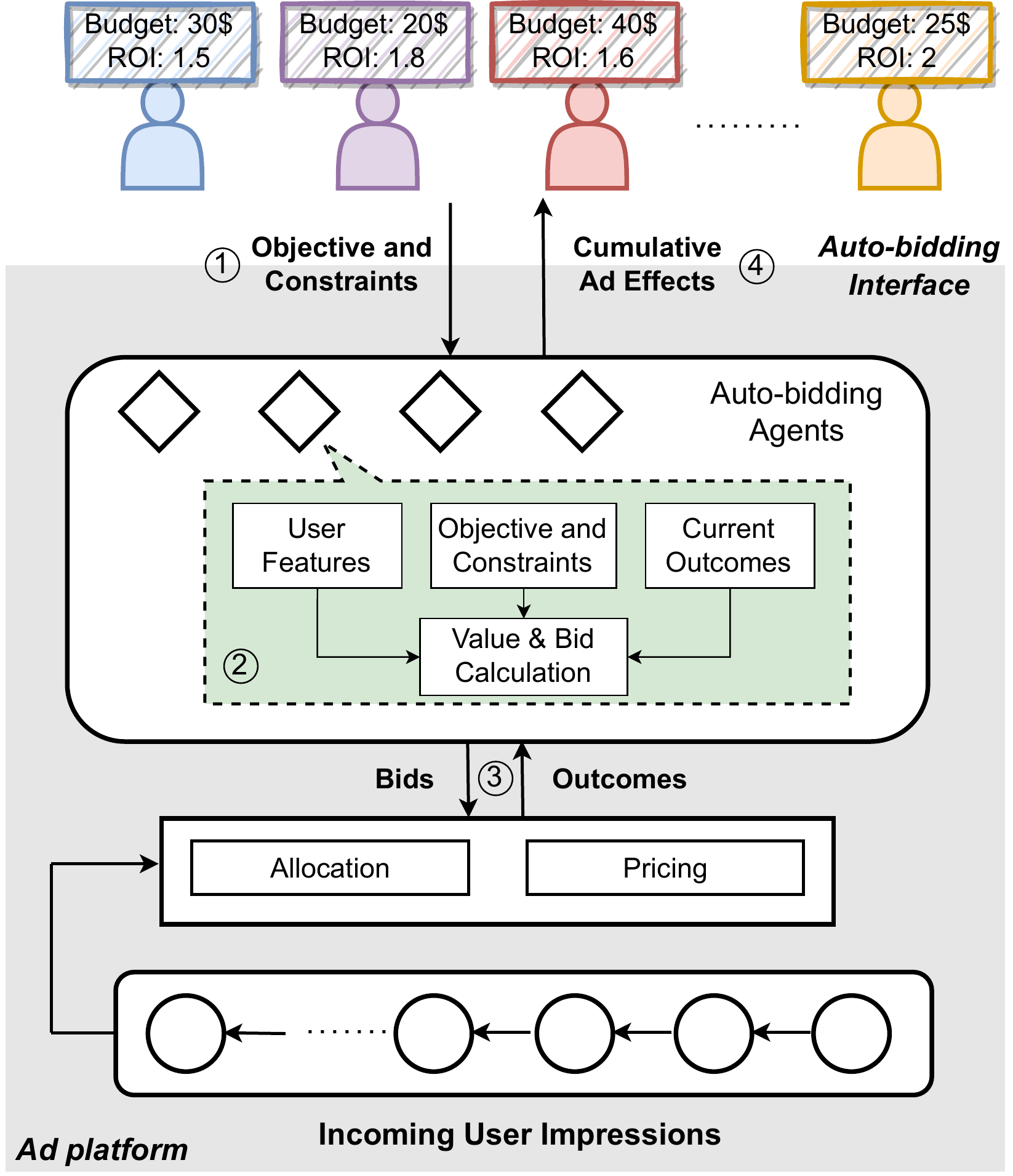}
		\caption{Ad Auction System Overview}
		\label{fig:ad-system}
	\end{figure}
	
	\vspace{-5pt}
	
	\subsection{Auction Model} \label{sec:model}
	Based on the motivations of the above online advertising system with auto-bidding, we propose the formal auction model considered in this work. There are $n$ advertisers competing for $m$ items (user impressions) coming in sequence during a time period with $m$ time slots, where only one item would appear in each time slot.\footnote{We would use advertisers with bidders, and items with impressions interchangeably throughout the work.} The advertisers are \emph{value-maximizing} bidders \cite{fadaei2016truthfulness,balseiro2022optimal,mehta2022auction}, who care about the cumulative value of her allocated impressions across all the time slots when the payment is within their financial constraints.  
	Each advertiser $i$ has budget ($B_i\geq 0$) and ROI ($R_i>0$) constraints, which are private information and are also called as type $t_i =(B_i, R_i)$ in mechanism design literature. We denote the type profile of all the advertisers as $\boldsymbol{t}=(t_i)_{i=1}^{n}$ and the space of the type profile as $\mathcal{T}=\prod_{i=1}^{n} \mathcal{T}_{i}$ with $\boldsymbol{t}\in \mathcal{T}$ and $\mathcal{T}_{i}=\mathcal{B}_i \times \mathcal{R}_i$. We denote the reported type of the other bidders except bidder $i$ as $\boldsymbol{t}_{-i}=(t_1,\ldots , t_{i-1},  t_{i+1},\ldots , t_n)$ and $\mathcal{T}_{-i}=\prod_{k\neq i}^{n} \mathcal{T}_{k}$. We assume advertisers' valuations on items are public information to the platform. We use $v_{i,j}> 0$ to represent the advertiser $i$'s valuation on item $j$.
	
	\par After collecting the budget and ROI of all the bidders, the online platform employs some auction mechanism $(\mathcal{A}, \mathcal{P})$ to decide ad allocations and payments, where $\mathcal{A}$ denote a (randomized) allocation rule $\mathcal{A}:\mathcal{T}\rightarrow [0, 1]^{n\times m} $ and $\mathcal{P}$ denote a (randomized) payment rule $\mathcal{P}:\mathcal{T}\rightarrow \mathbb{R}^n$. Specifically, for a reported type profile $\boldsymbol{t}^\prime \in \mathcal{T}$, the probability of bidder $i$ being allocated item $j$ is denoted by $a_{i,j}(\boldsymbol{t}^\prime)$ and the expected payment of bidder $i$ is denoted by $p_i(\boldsymbol{t}^\prime)$.  For any item $j\in [m]$, the allocation constraint is $\sum _{i=1} ^{n}  a_{i,j}(\boldsymbol{t}^\prime)\leq 1$. Bidder $i$'s cumulative value in these auctions is 
    \vspace{-8pt}
 
    $$v_i(\boldsymbol{t}^{\prime})=\sum _{j=1} ^{m} v_{i, j} a_{i,j}(\boldsymbol{t}^\prime)$$ and her realized ROI is $\operatorname{ROI}_i(\boldsymbol{t}'):= v_i(\boldsymbol{t}')/p_i(\boldsymbol{t}')$ (considered as $+\infty$ if $p_i(\boldsymbol{t}')=0$). 
	
	\par The utility for a bidder with budget and ROI constraints and true type $t_i$ when reporting $t_i'$ is defined as
        \vspace{-10pt}
 
	\begin{equation*}
		u _{i}\left(t_{i}, \boldsymbol{t}'\right)=\left\{\begin{aligned}
			v _i(\boldsymbol{t}'), & \text { if } p_i(\boldsymbol{t}') \leq B_{i} \text { and } \operatorname{ROI}_i(\boldsymbol{t}') \geq R_i, \\
			-\infty, & \text {  otherwise, }
		\end{aligned}\right.
	\end{equation*}
	where the type profile $\boldsymbol{t}'=(t_i', \boldsymbol{t}_{-i}')$.
	\par In this work, we focus on designing \emph{truthful} auction mechanisms satisfying incentive compatible (IC) and individual rationality (IR) conditions for budget and ROI, which is a two-dimensional mechanism design problem. Various multi-dimensional mechanism design problems are shown to be difficult in both analytical and computational aspects \cite{pavlov2011optimal,chen2014complexity,daskalakis2015multi}.
	In particular, we consider \textit{dominant-strategy incentive compatible} (DSIC) and \textit{individual rational} (IR) direct-revelation mechanisms.
	\begin{definition}\label{def:DSIC}
		An auction mechanism is dominant strategy incentive compatible (DSIC) if $\forall i \in [n]$, $t_i, t_i'\in \mathcal{T}_i, \boldsymbol{t}_{-i}\in \mathcal{T}_{-i}$: $u_i(t_i, (t_i, \boldsymbol{t}_{-i}))\geq u_i(t_i, (t_i', \boldsymbol{t}_{-i}))$.  
	\end{definition}
	
	\begin{definition}\label{def:IR}
		An auction mechanism is individual rational (IR) if $\forall \boldsymbol{t}\in \mathcal{T}$, $i \in [n]$: $p_i(\boldsymbol{t})\leq B_i$ and $\operatorname{ROI}_i(\boldsymbol{t})\geq R_i$.
	\end{definition}
	
	\par The online platform typically has some objectives to maximize. Two common design objectives are \emph{revenue} and \emph{social welfare}. Revenue is defined as the sum of payment from bidders, \emph{i.e.}, $\sum _i p_i(\boldsymbol{t})$. For social welfare, we need a synonymous metric, \emph{liquid welfare}, defined as the maximum revenue that can be extracted from bidders without breaking IR constraints, to incorporate the existence of financial constraints \cite{azar2017liquid,aggarwal2019autobidding}. In our context, liquid welfare can be defined as 
	\begin{equation} \label{eqn:liquid-welfare}
	\operatorname{LW} = \sum_{i \in [n],u_i(t_i, \boldsymbol{t}')\geq 0} \min \left( \frac{v_i(\boldsymbol{t}')}{R_i}, B_i\right).
	\end{equation}

	\section{Characterization of Truthfulness}\label{sec:IC}
        \subsection{Conditions for Truthfulness}
	In this section, we investigate the	truthfulness conditions for this multi-dimensional mechanism design problem. The detailed proofs of our results are presented in the appendix. 
	\par To provide an intuition about the differences between auction models with private constraint and private valuation, we start from the analysis of one private constraint. 
    To satisfy the IC property on budget, we need to guarantee the bidders for not obtaining a higher utility by reporting a smaller or larger budget. 
	As reporting a smaller budget will not lead a bidder to break her original budget constraint, the gained utility should decrease for reducing the budget. For the bidder misreporting a larger budget and obtaining higher values, we need to charge the bidder to break her original budget constraint, resulting in negative infinite utility to prevent misreporting.
 
	\begin{theorem}\label{thm:IC-B}
		An auction mechanism is DSIC on budget $B$ only if $\forall i \in [n]$, $R\in \mathcal{R}_i$, $\boldsymbol{t}_{-i}\in \mathcal{T}_{-i}$: \\(1) $v_i\left((B, R), \boldsymbol{t}_{-i}\right)$ is non-decreasing in $B$; \\(2) If $v_i\left((B', R), \boldsymbol{t}_{-i}\right)> v_i\left((B, R), \boldsymbol{t}_{-i}\right)$ for $B'>B$, then 
        \vspace{-6pt}
  
        $$p_i\left((B', R), \boldsymbol{t}_{-i}\right)>B.$$
	\end{theorem}
	
	We can form a similar statement for ROI by showing that there is no incentive for a bidder to misreport her ROI.
	
	\begin{theorem}\label{thm:IC-R}
		An auction mechanism is DSIC on ROI $R$ only if $\forall i \in [n]$, $B\in \mathcal{B}_i$, $\boldsymbol{t}_{-i}\in \mathcal{T}_{-i}$:\\(1) $v_i\left((B, R), \boldsymbol{t}_{-i}\right)$ is non-increasing in $R$; \\(2) If $v_i\left((B, R'), \boldsymbol{t}_{-i}\right)> v_i\left((B, R), \boldsymbol{t}_{-i}\right)$ for $R'<R$, then
        \vspace{-6pt}
  
        $$\operatorname{ROI}_i\left((B, R'), \boldsymbol{t}_{-i}\right)<R.$$
	\end{theorem}
	
	\par \noindent Although the monotone properties required by conditions (1) in Theorem \ref{thm:IC-B} and \ref{thm:IC-R} seem similar to the truthfulness conditions for quasi-linear bidder with private valuations, \emph{e.g.}, Myerson's Lemma~\cite{myerson1981optimal}, conditions (2) reveal their basic differences: assigning a payment to break at least one constraint (lead to negative infinite utility) is indispensable to prevent misreporting for private constraints.
 \par Theorem \ref{thm:IC-B} and \ref{thm:IC-R} naturally provide necessary conditions on the DSIC of two-dimensional type $(B, R)$. The following result shows that these conditions 
	also provide the sufficient conditions. That is, if a bidder cannot obtain higher utility through misreporting one of her constraints, misreporting the two constraints can also not obtain higher utility.
	
	\begin{theorem}\label{thm:IC-BR} 
		An auction mechanism is DSIC on both budget and ROI if and only if it satisfies Theorem \ref{thm:IC-B} and \ref{thm:IC-R}.
	\end{theorem}

    Theorem \ref{thm:IC-BR} is proved through showing none of the misreporting may bring the bidder higher utility when conditions in Theorem \ref{thm:IC-B} and \ref{thm:IC-R} hold, which relies on the mathematical relationship between payment and ROI. In Theorem \ref{thm:IC-BR}, the truthful conditions in Theorem \ref{thm:IC-B} for budget appear to be independent with the truthful conditions in Theorem \ref{thm:IC-R} for ROI, and satisfy these conditions simultaneously could achieve the two-dimensional truthfulness.
    Nevertheless, the payment term actually appear in both the conditions (recall $\operatorname{ROI}=v/p$ in condition (2) of Theorem \ref{thm:IC-R}), \emph{i.e.}, we have to use the same payment scheme to satisfy these two sets of conditions. 
    In order to step toward the full characterization of truthfulness, we need to further analyse how payment is influenced and constrained by financial constraints and allocation. 
    \par To facilitate our discussion, we define an allocation rule $\mathcal{A}$ to be \textit{feasible} if there exists a payment rule $\mathcal{P}$ such that the mechanism $(\mathcal{A},\mathcal{P})$ is truthful. The set of feasible allocation rule is the space we can search for truthful auctions. Although there may exist multiple payment rules that constitute a truthful auction for a feasible allocation rule $\mathcal{A}$, the following theorem allows us to focus on the maximum payment rule.
	
	 \begin{theorem}\label{thm:BR-vm-payment}
		 For a feasible allocation rule $\mathcal{A}$, assigning
		 \begin{equation} \label{eqn:ic-payment}
		 	p_i\left((B_i,R_i), \boldsymbol{t}_{-i}\right)=\min \left( v_i\left((B_i,R_i), \boldsymbol{t}_{-i}\right)/R_i, B_i\right),
		 \end{equation}
		  for $\forall i\in [n], (B_i, R_i)\in \mathcal{T}_i$, $ \boldsymbol{t}_{-i}\in \mathcal{T}_{-i}$ constitutes a truthful auction mechanism.
		 \end{theorem}
	
	\par Theorem \ref{thm:BR-vm-payment} naturally holds because each type is charged the maximum payment under her financial constraints. If some other payment could constitute a truthful auction with $\mathcal{A}$, then improving the payment to this maximum amount would not break the originally held conditions (2) in Theorem \ref{thm:IC-B} and Theorem \ref{thm:IC-R}. By Theorem \ref{thm:BR-vm-payment}, given any feasible allocation rule $\mathcal{A}$, we could find a corresponding payment rule to constitute a truthful auction. 
    As the bidders do not evaluate detailed allocation results of each individual auction, searching for the truthful allocation rule is equal to designing cumulative value functions $v_i((B_i, R_i), \textbf{t}_{-i})$ that can be realized. Thus, we could use the corresponding cumulative value function to represent a feasible allocation rule. 
    \par Our next step is to further shrink the consideration space of comparison in conditions of truthfulness. Previous conditions in Theorem \ref{thm:IC-BR} include comparisons between any two type sharing the same budget or ROI, \emph{e.g.}, comparing any $(B, R)$ and $(B', R)$ in Theorem \ref{thm:IC-B}, which is still a large consideration space and leads to difficulties in finding the cumulative value functions. To enable analysis between ``neighbouring'' types, for any allocation rule $\mathcal{A}$ and its cumulative value function $v(\boldsymbol{t})$, we assume the limits  
    $\lim _{B\to B_i^{-}}v_i\left((B, R_i), \boldsymbol{t}_{-i}\right)\text{ and }\lim _{R\to R_i^{+}}v_i\left((B_i, R), \boldsymbol{t}_{-i}\right)$
    exist for any $\boldsymbol{t}\in \mathcal{T}, i\in [n]$. 
    Through considering these infinitely close types and substituting the payment terms in Theorem \ref{thm:IC-BR} by the payment rule (\ref{eqn:ic-payment}), conditions that constrain the payment or realized ROI of misreporting another type to break the original constraints could be converted to conditions regarding cumulative value of the type itself.
	
	\begin{theorem}\label{thm:BR-vm-vrange}
		An allocation rule $\mathcal{A}$ can derive a  truthful auction if and only if $\forall \boldsymbol{t}\in \mathcal{T}$, $i\in [n]$:\\ (1) The cumulative value $v_i\left((B, R), \boldsymbol{t}_{-i}\right)$ is non-decreasing in $B$ for $R=R_i$ and non-increasing in $R$ for $B=B_i$; \\ (2) If $v_i\left((B_i,R_i), \boldsymbol{t}_{-i}\right) > \lim _{B\to B_i^{-}}v_i\left((B, R_i), \boldsymbol{t}_{-i}\right)$, then 
        \vspace{-6pt}
        
        $$v_i\left((B_i,R_i), \boldsymbol{t}_{-i}\right) \geq B_i \times R_i;$$ 
        (3) If $v_i\left((B_i,R_i), \boldsymbol{t}_{-i}\right) > \lim _{R\to R_i^{+}}v_i\left((B_i, R), \boldsymbol{t}_{-i}\right)$, then 
        \vspace{-6pt}
        
        $$v_i\left((B_i,R_i), \boldsymbol{t}_{-i}\right) \leq B_i \times R_i.$$
	\end{theorem} 

 \subsection{Structures of Feasible Allocation Rule}
 Theorem \ref{thm:BR-vm-vrange} fully characterizes the conditions of a feasible allocation rule. In this subsection, we would further exploit the structures of feasible allocation rule indicated by Theorem \ref{thm:BR-vm-vrange} to provide more instructions on truthful auction design.
 
 \par By Theorem \ref{thm:BR-vm-vrange}, the relation between cumulative value $v$ and the term $B\times R$ strictly characterizes whether this type shares the same cumulative value with its neighbouring types, which allows us to further clearly represent  the structure of $v$. Fixing the budget, with the decrease of ROI, the cumulative value assigned to the type should be increasing, notice the term $B\times R$ decreases along with $R$; however, condition (3) in Theorem \ref{thm:BR-vm-vrange} requires that if the cumulative value strictly increases in this process, then its value should not exceed $B\times R$. Thus, there must exist some threshold ROI, such that the increasing cumulative value intersects with the decreasing $B\times R$, and the cumulative value could not increase anymore. Based on the above observation, we define 
 \vspace{-12pt}
 
 $$\operatorname{thr}_i(B, \boldsymbol{t}_{-i})=\sup _{R \in \mathcal{R}_i} \left( v_i\left((B, R), \boldsymbol{t_{-i}}\right)\geq B\times R \right)$$ 
 if the considered set is non-empty and bounded, or otherwise be 0.\footnote{This set would be empty only when $v_i\left((B, R), \boldsymbol{t_{-i}}\right)=0$ for all $R\in \mathcal{R}_i$, and unbounded only when $B=0$.} It turns out that the cumulative value assigned to this threshold ROI is exactly $B\times R$. 

		\begin{figure}[t]
		\centering
		\includegraphics[width=0.8\linewidth]{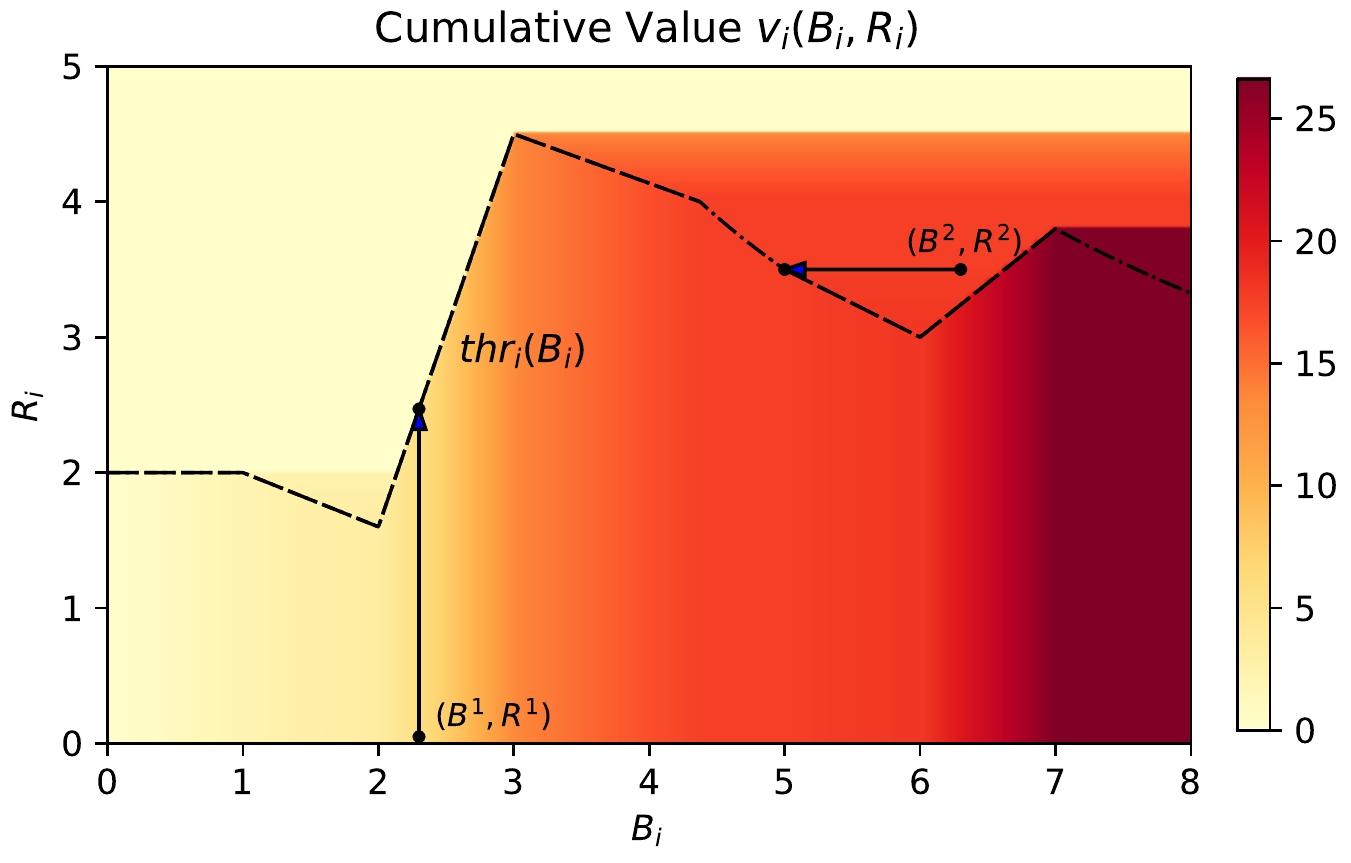}
		\caption{An example of mapping from a threshold ROI function to cumulative values fixing the others' types $\boldsymbol{t}_{-i}$.}
		\label{fig:pooling-1}
    \vspace{-8pt}
	\end{figure}

	\begin{theorem} \label{thm:column-structure}
		$\forall i\in [n]$, $\boldsymbol{t}_{-i} \in \mathcal{T}_{-i}$, $B_i\in \mathcal{B}_i$, $R_i \leq \operatorname{thr}_i(B_i, \boldsymbol{t}_{-i})$: 
    \vspace{-8pt}

  $$ v_i\left((B_i, R_i), \boldsymbol{t_{-i}}\right)= \operatorname{thr}_i(B_i, \boldsymbol{t}_{-i})B_i.$$
	\end{theorem}
	
    We could then combine different budgets together to obtain the complete characterization of feasible allocation in the two-dimensional type space. Since all types above the threshold ROI has cumulative value $v<B\times R$ by the definition of $\operatorname{thr}$ function, they could not satisfy the condition (2) in Theorem \ref{thm:BR-vm-vrange}, and thus are forced to have the same cumulative value with their neighbouring types with a smaller budget. 
 
	\begin{theorem} \label{thm:row-structure}
		$\forall i\in [n]$, $\boldsymbol{t}_{-i} \in \mathcal{T}_{-i}$, $B_i\in \mathcal{B}_i\backslash\{0\}$, $R_i > \operatorname{thr}_i(B_i, \boldsymbol{t}_{-i})$: 
        \vspace{-6pt}
  
        $$v_i((B_i, R_i), \boldsymbol{t}_{-i})=\lim _{B\to B_i^{-}}v_i\left((B, R_i), \boldsymbol{t}_{-i}\right).$$
	\end{theorem}
 
	Since the types below (Theorem~\ref{thm:column-structure}) and above (Theorem~\ref{thm:row-structure}) the ROI threshold have both been fully described, the cumulative values on the entire type space could be fully defined if $\operatorname{thr}_i(B_i,\boldsymbol{t}_{-i})$ is given and $\boldsymbol{t}_{-i}$ is fixed.  We provide an example in Figure \ref{fig:pooling-1} to illustrate these characterizations. 
 For types below the threshold ROI curve and have the same budget $B$ (\emph{i.e.}, types lie in the same vertical line), they share the same cumulative value $\operatorname{thr}_i(B)\times B$. For each type above the curve, its cumulative value is equal to that of its neighbouring type with a smaller budget until it reaches some type on the threshold ROI curve. 
 This is equivalent to search horizontally to the left to find the closest type on its corresponding ROI threshold, where all the types along this horizontal search share the same cumulative value.

    \par We have shown that each feasible two-dimensional cumulative value function (whose input is $(B, R)$) for a single bidder could be represented by a one-dimensional threshold ROI function (whose input is $B$) when other bidders' profile is fixed, but it is unknown what kinds of one-dimensional threshold ROI functions could represent a feasible cumulative value function. By condition (1) in Theorem \ref{thm:BR-vm-vrange}, the cumulative value when $R$ is infinitely close to $0$ should be increasing for different $B$, which indicates $\operatorname{thr}_i(B, \boldsymbol{t}_{-i}) \times B$ needs to be non-decreasing in $B$. 
    Since every cumulative value function mapped from a non-decreasing $\operatorname{thr}_i(B, \boldsymbol{t}_{-i}) \times B$ satisfies Theorem \ref{thm:BR-vm-vrange}, this becomes the only requirement for a feasible $\operatorname{thr}_i$ function. 
    For convenience in notations, we denote $g_i(B_i, \boldsymbol{t}_{-i}) =\operatorname{thr}_i(B_i, \boldsymbol{t}_{-i}) \times B_i$. 
We denote $\mathcal{V}$ as the space of feasible cumulative value function  when $i$ and $\boldsymbol{t}_{-i}$ are fixed, \emph{i.e.}, $v_i:T_i \to \mathbb{R}^+$, and denote $\mathcal{G}$ as the space of non-decreasing function $g: \mathbb{R}^+ \to \mathbb{R}^+$ with $g(0)=0$. 
 
	\begin{corollary} \label{col:full-chara}
		There exists a bijective map $m: \mathcal{G} \to \mathcal{V}$ with 
        \vspace{-6pt}
  
  		\begin{equation*}
			(m \circ g) (B, R)=\left\{
			\begin{aligned}
				&g(B) \text{\quad \quad \quad \quad \quad \quad \quad \quad if } R\leq \frac{g(B)}{B},& \\
				& g\left(\sup_{B'\leq B} \{R\leq \frac{g(B')}{B'}\} \right) \text{\quad otherwise,}
			\end{aligned}
			\right.
		\end{equation*}
  where we define $\sup \varnothing = 0$.
	\end{corollary}

In other words, given any one-dimensional non-decreasing function with $g(0)=0$, we could transform it to a feasible\footnote{Our feasibility is with respective to truthfulness, and does not specify how to realize the allocations from detailed items.} cumulative value function by the above mapping process.

\section{Truthful Auction Design} \label{sec:mechanism-design}
    The derived structure of feasible allocation rules give rise to a \emph{value grouping phenomenon}. 
     When reversing the process of finding the cumulative value of a certain type given threshold ROI curve, every type on the threshold ROI curve shares the same cumulative value with two groups of types: the types vertically below it,  as well as the types locate horizontally right to it and above its corresponding threshold ROI, as shown in Figure \ref{fig:pooling-1}.
     This intricate value grouping phenomenon  are enforced by the truthful conditions instead of the optimization requirement. As we can observe in Figure \ref{fig:pooling-1}, since the horizontal search terminates when meeting a type on the threshold ROI curve, the  grouping shape is determined by the relative rank of $thr_i(B_i, \boldsymbol{t}_{-i})$ for different $B_i$. However,  though $thr_i(B_i, \boldsymbol{t}_{-i})$ is required to keep $g(B)$ non-decreasing, it is not required to be monotone itself, which leads the terminate of horizontal search and the grouping shapes to be irregular. This also brings difficulties in deriving general closed-form conditions for the group of types to share the same allocation. Furthermore, although the cumulative values of the grouping types are forced to be the same, their payment are calculated by Theorem \ref{thm:BR-vm-payment}, leading to varying payments between grouping types and the non-linear revenue with respective to the threshold ROI. 
	\par The above characteristics of value grouping phenomenon, \emph{i.e.}, irregularity in grouping shape and non-linear payment, make any analytical or mathematical programming approach difficult to adopt. To reach feasible and implementable auction design, we propose a family of simple truthful auctions based on newly designed rank score functions. 

	\begin{algorithm}[t]
		\caption{A Family of Simple Truthful Auctions}
		\label{alg:dsic-fami}
		\KwIn{Bidder's reported type $(B_i, R_i)$, and non-increasing rank score functions $f_{i,j}$.}
		\KwOut{Bidder's allocated items $A_i$ and payment $p_i$.}
		Initialize $A_i=\{\}$ for $i\in [n]$; \\
		Compute virtual bids $b_{i,j} \gets v_{i,j}\times f_{i,j}(R_i)$; \\
		\For{each item $j\in [m]$}{
			Find bidder $i_0$ with the highest virtual bid $b_{i_0, j}$; \\
			Record the second highest bid $c_{j} \gets  \max _{i\neq i_0} b_{i,j}$; \\
			Add item $j$ into $A_{i_0}$ and set  $a_{i_0, j}=1$; \\
		}
		\For{each bidder $i \in [n]$}{
			Compute ROI $r_{i,j}=f_{i,j}^{-1}(\frac{c_j}{v_{i,j}})$ required for winning the item $j\in A_i$; \\
			Find a largest $R_i^c \in \mathcal{R}_i$ s.t. $\sum _{j\in A_i} \mathbb{I}\left(R_i^c \leq r_{i,j} \right) \times \frac{v_{i,j}}{R_i^c} \geq B_i $, and let \quad \quad $d \gets $ the difference between left-hand side and right-hand side of the above inequality;\\
			\If{$d>0$} {Find the items $j'$ with $r_{i,j'}=R_i^c$, and remove the items with value $d/R_i^c$; \\
			}		
			\If{$R_i \leq R_i^c$}{
				Remove the items with $r_{i,j} < R_i^c$ from $A_i$;\\
			}
			$p_i \gets \min \left(\frac{\sum _{j\in A_i} v_{i,j}\cdot a_{i,j}}{R_i}, B_i\right);$
		}
		\textbf{return} $(A_i, p_i)$ for $i\in [n]$.
	\end{algorithm}
 
	\par The detailed auction mechanism is presented in Algorithm~\ref{alg:dsic-fami}. For each item, bidders are ranked based on their \emph{virtual bids} $b_{i,j}$, which is defined as $v_{i,j}$ multiplying a rank score $f_{i,j}(R_i)$ with some pre-defined non-increasing function $f_{i,j}$ (Line 2). The items are then temporarily allocated to each bidder with the highest virtual bids (Line 4-6). After the candidate allocation set $A_i$ for each bidder $i$ has been determined, we compute the the corresponding ROI requirement $r_{i,j}$ for the bidder $i$ to rank the first in item $j\in A_i$, \emph{i.e.}, bidder $i$ needs to report $R_i\leq r_{i,j}$ in order to maintain ranked first for item $j$ (Line 8). The critical ROI in Line 9 is our key design to guarantee truthfulness, which computes the largest ROI a bidder could report to win enough items in $A_i$ and spend out her budget. 
    Intuitively, critical ROI simulates the best-response ROI of the bidder given her budget constraint, which does not involve and thus keeps independent of her true ROI constraint. 
    We naturally utilize this critical ROI, which transforms the budget to the same dimension of ROI, as the threshold ROI function of our auction, and design the remaining parts based on the truthful conditions in Section \ref{sec:IC}. Line 10-11 are designed to guarantee the value of remaining items in $A_i$ would exactly equal to the budget of the bidder times the computed critical ROI (Theorem \ref{thm:column-structure}). Line 12-13 are designed to guarantee the types with ROI larger than the corresponding critical ROI would be allocated based on their reported ROI (Theorem \ref{thm:row-structure}) and get all items with $r_{i,j} \geq R_i$. Line 14 determines the payment for each bidder $i$ based on Theorem \ref{thm:BR-vm-payment}. Due to its low time complexity, Algorithm \ref{alg:dsic-fami} has good implementation scalability.
	
	\begin{theorem}\label{thm:dsic-fami}
		If the pre-determined functions $f_{i,j}(R)$ are non-increasing with $R$, then the auction mechanism corresponding to Algorithm \ref{alg:dsic-fami} is truthful.
	\end{theorem}

    \noindent Our proof is conducted through figuring out and verifying the corresponding $g(B)$ and threshold ROI function of Algorithm~\ref{alg:dsic-fami}. We can design personalized $f_{i,j}(R)$ for bidders and items according to the auction objective and prior knowledge before the advertising campaign, such as the type distribution of bidders and the information of incoming items. As the cost of bypassing the design hardness brought by value grouping of multiple bidders, our mechanism may not allocate all the items if the corresponding highest ranking bidder has fully consumed her budget. Nevertheless, the auctioneer could take multiple approaches to alleviate this phenomenon. While adjusting the functions $f_{i,j}(R)$ with the pre-known information is the most direct approach, dividing advertisers into groups and circulating the unsold items could also be effective. 
	
	\section{Experiments}
	In this section, we conduct experiments with synthetic data to validate the performance of our proposed auctions.
	
	\subsection{Experimental Setup}
	\par There are two sets of experiments with i.i.d. and non-i.i.d. bidders, respectively.
	We vary the number of bidders and items to simulate the changes in demand and supply. The presented results are averaged by 50 runs. We choose distribution parameters based on the common fluctuations in advertising markets, and normalize the lower bound of $v_{i,j}$. \\
	\textbf{Symmetric bidders } Bidders are symmetric with $v_{i,j}\sim U[1,4]$, $B_i\sim U[40, 80]$ and $R_i\sim U[1,3]$, where $U[a,b]$ is the uniform distribution within the range $[a,b]$.\footnote{We have tested several other parameters under uniform distribution, and the trends of results are the same.} \\ 
    	\textbf{Mixed bidders } There are low and high distributions for $v_{i,j}$, $B_i$ and $R_i$ as bidders' possible types. Specifically, $v_{i,j}\sim U[1,2]$ or $v_{i,j}\sim U[2,3]$, $B_{i}\sim U[20,40]$ or $B_{i}\sim U[80,100]$, and $R_{i}\sim U[1,2]$ or $R_i\sim U[2,3]$. 
	Combinations of these categories result in 8 groups of bidders. \\
	\noindent \textbf{Baseline Auctions }
	Since no existing mechanism guarantees the IC properties of both budget and ROI, we consider common repeated auction formats: first-price and second-price auctions. Due to their non-truthfulness (examples provided in the appendix), we involve misreporting for these auctions.\\
	$\bullet$ \emph{Repeated first-price auctions: }
    The auctioneer holds first-price auctions for every single item, where bidder $i$ bids $v_{i,j}/R_i$ for item $j$. The auctioneer allocates the item to the highest ranking bidder who has remaining budget to afford, and charge the first-price payment. We aim to simulate bidders' misreporting behaviors by the classical \emph{best-response} dynamics with the true profile as starting points. In repeated first-price auctions, the bidders have no incentive to misreport their budget. We calculate the best response ROI as the smallest ROI that wins the most items in historical auctions without breaking the original financial constraints.
	
	\noindent $\bullet$ \emph{Repeated second-price auction:} This auction format is similar to repeated first-price auction with second-price auction used in each round instead. 
	We still simulate bidders' misreporting behaviors to be the best-response dynamics as above, with payment changed to be the second price, where the best-response ROI here may be larger or smaller than the true ROI.\\ 
	$\bullet$ \emph{Non-truthful Optimal Baseline:} The optimal offline revenue could be computed through linear programming, where we set each bidder $i$'s payment to item $j$ as exactly $v_{i,j}\times a_{i,j}/R_i$, and her total payment is constrained to not exceed $B_i$. \\
	\noindent \textbf{Evaluation Metrics} We consider revenue and liquid welfare as the optimization goal. Since the payment formula~(\ref{eqn:ic-payment}) in Theorem~\ref{thm:BR-vm-payment} aligns with the definition of liquid welfare in equation~(\ref{eqn:liquid-welfare}) when no bidder receives negative infinite utility, the metrics of revenue and liquid welfare would be the same for truthful auctions, and we will thus only report revenue in our results. Besides, as advertisers may have fluctuating auction performances across time slots, \emph{fairness} should also be considered in order to preserve bidders' willingness to attend the auctions. 
    Using liquid welfare to substitute the traditional valuation function \cite{bezakova2005allocating,chakrabarty2009allocating}, fairness is defined as 
    \vspace{-8pt}
 
    $$\operatorname{fairness}=\min _{i\in [n]} \min \left( \frac{v_i(\boldsymbol{t}')}{R_i}, B_i\right).$$
     \\
	\noindent \textbf{Rank Score Function } We adopt rank scores in the form $f_{i,j}(R)=\alpha _{i,j} \times e^{-\beta R}$, where $\alpha _{i,j}$ is drawn from a rectified normal distribution $N(\mu _i, \sigma_i ^2)$, and $\beta, \mu _i, \sigma_i$ are pre-set parameters. In $f(R)$, $\beta$ is used to adjust the impact of ROI on equivalent bids, and $\alpha$ maintains the rank scores of various bidders in a comparable range. Since items may be discarded when some bidders win excessive items in the proposed auction (Line 13 in Algorithm \ref{alg:dsic-fami}), in order to avoid loss in welfare and revenue, we should give other bidders non-trivial winning opportunities when the supply is sufficient, which is provided by the randomness in $\alpha \sim N(\mu _i, \sigma_i ^2)$. For each automated bidding environment with different setting of bidders and items, we set rank score function parameters to be the same for bidders following the same distribution, and choose the parameters with better revenue in this environment.
	
	\subsection{Experimental Results}
	\begin{figure*}[ht]
		\centering
		\begin{subfigure}[t]{0.24\textwidth}
			\hspace{-5pt}
			\centering
			\includegraphics[width=1\linewidth]{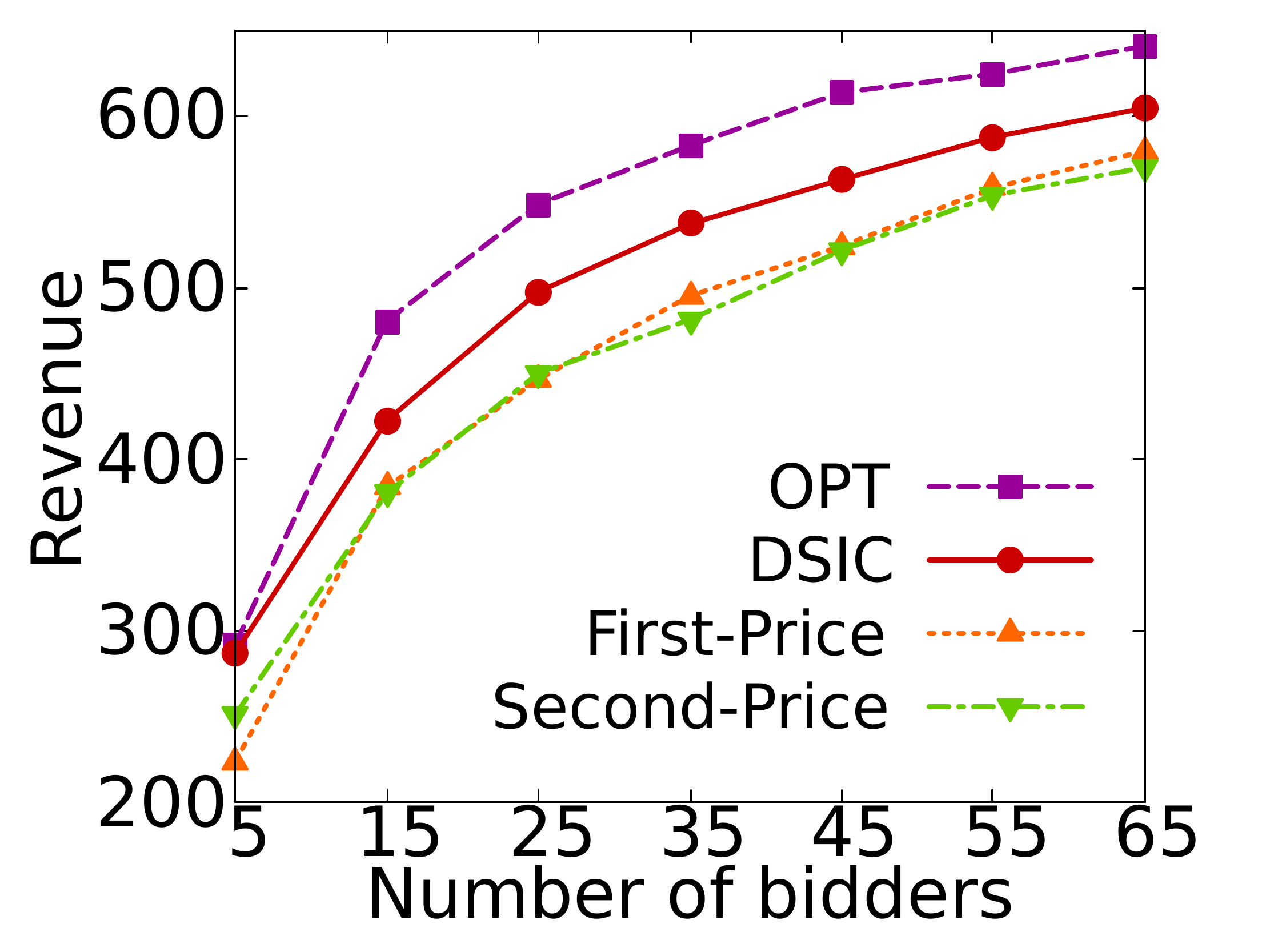}
			\caption{}
			\label{fig:bnum-vary}
		\end{subfigure}
		\begin{subfigure}[t]{0.24\textwidth}
			\centering
			\includegraphics[width=1\linewidth]{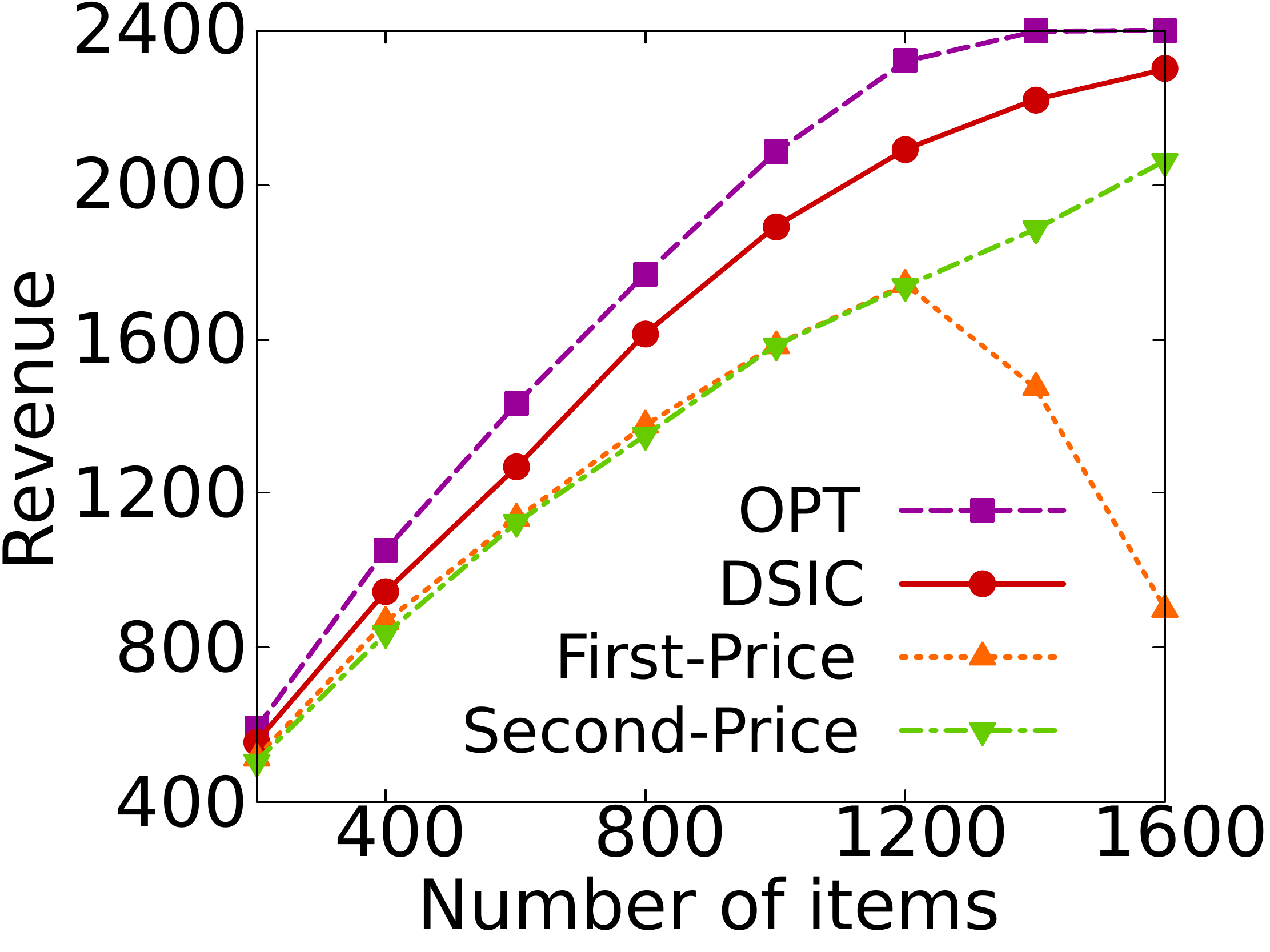}
			\caption{}
			\label{fig:inum-vary}
		\end{subfigure}
		\centering
		\begin{subfigure}[t]{0.24\textwidth}
			\centering
			\includegraphics[width=1\linewidth]{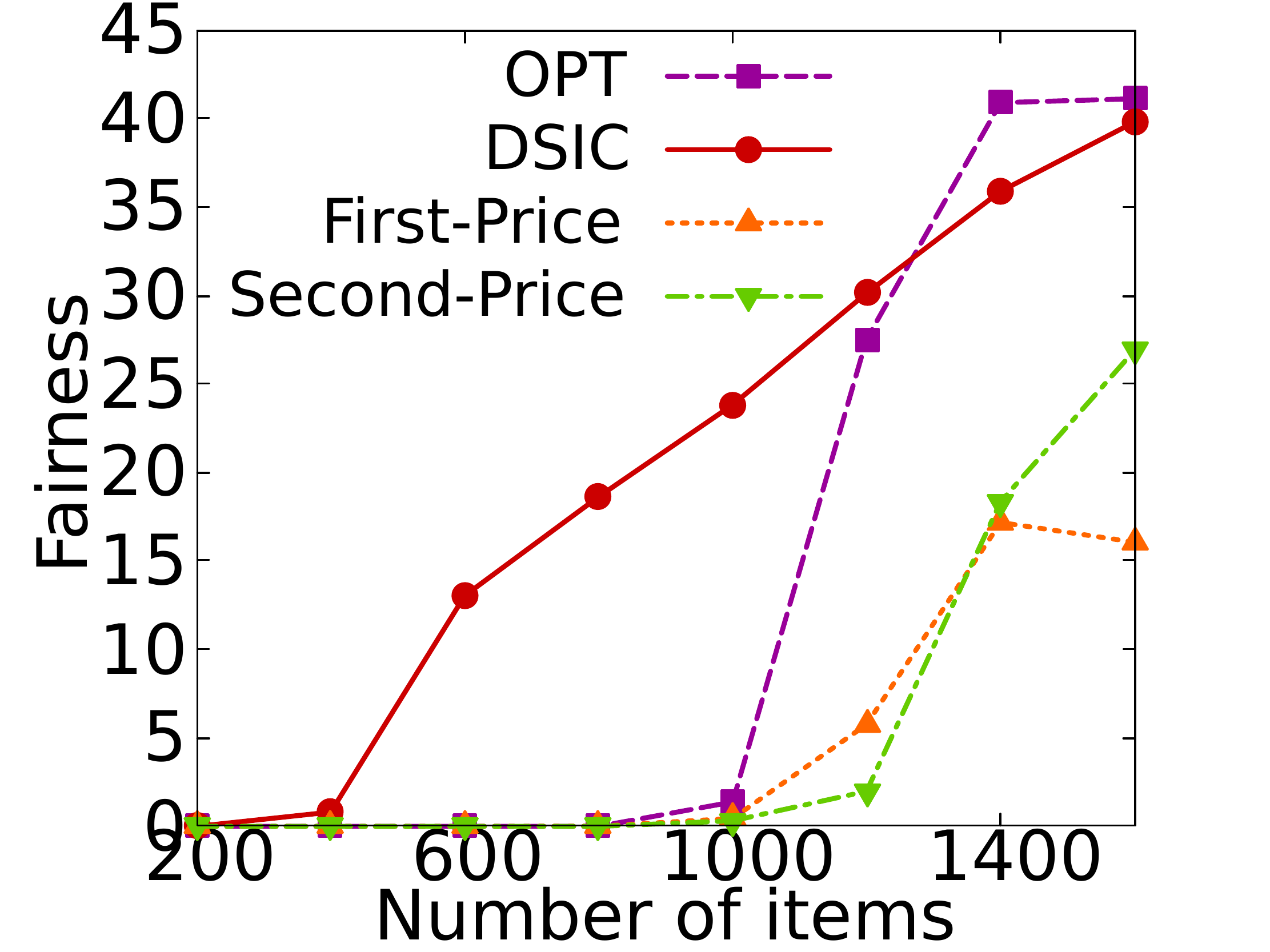}
			\caption{}
			\label{fig:inum-fair}
		\end{subfigure}
		\begin{subfigure}[t]{0.24\textwidth}
			\hspace{-5pt}
			\centering
			\includegraphics[width=1\linewidth]{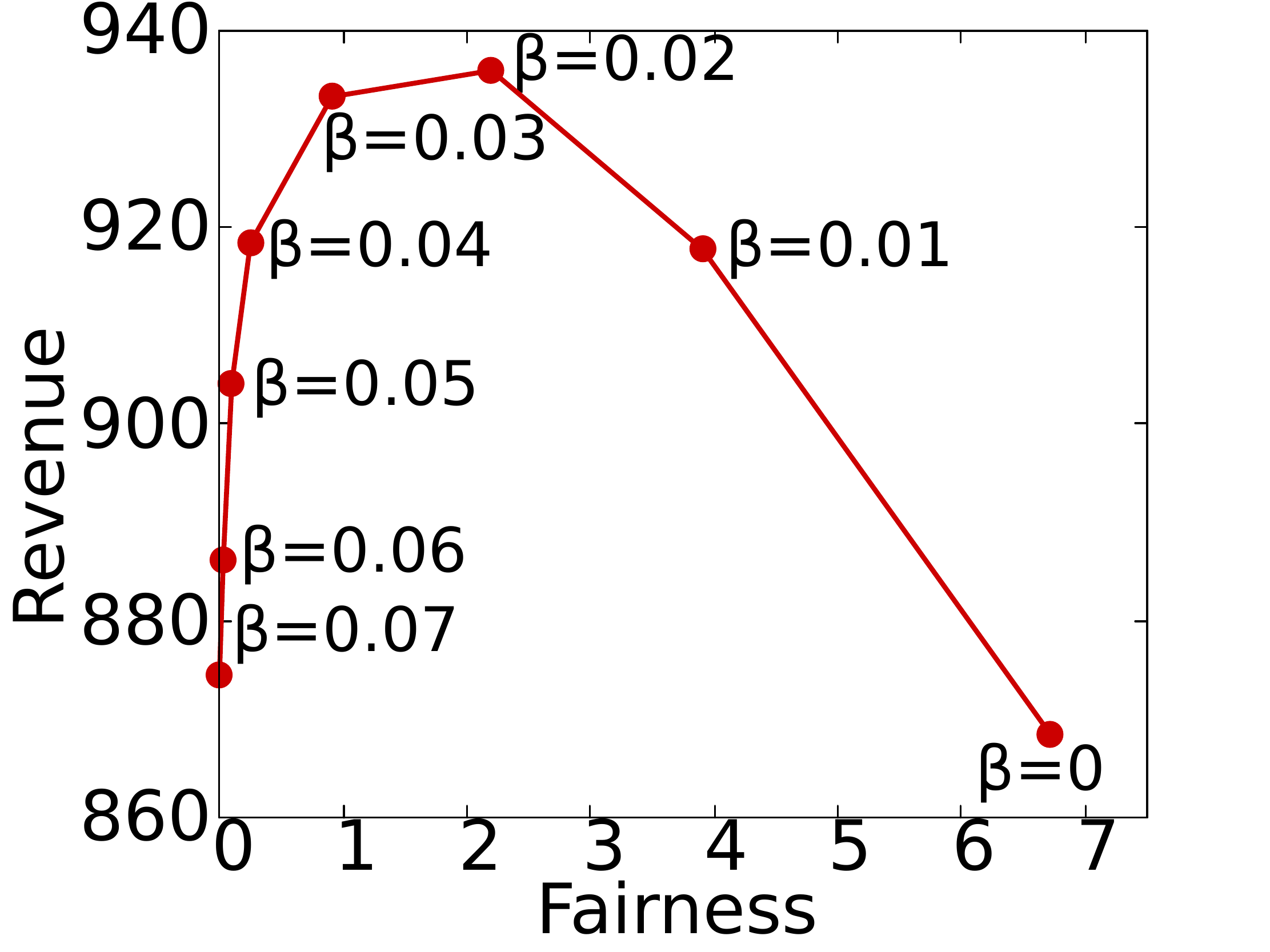}
			\caption{}
			\label{fig:fair-trade}
		\end{subfigure}
		\caption{Revenue and fairness in different experiment setting for symmetric bidders: (a) Revenue for different number of i.i.d bidders with 200 items; (b) Revenue for different number of items with 40 i.i.d bidders; (c) Fairness for different number of items with 40 i.i.d bidders; (d) Fairness of our DSIC auction with 40 i.i.d bidders and 400 items using different parameter $\beta$ in rank score functions;}
	\end{figure*}
	
	\begin{figure}[t]
		\centering
		\includegraphics[width=0.7\linewidth]{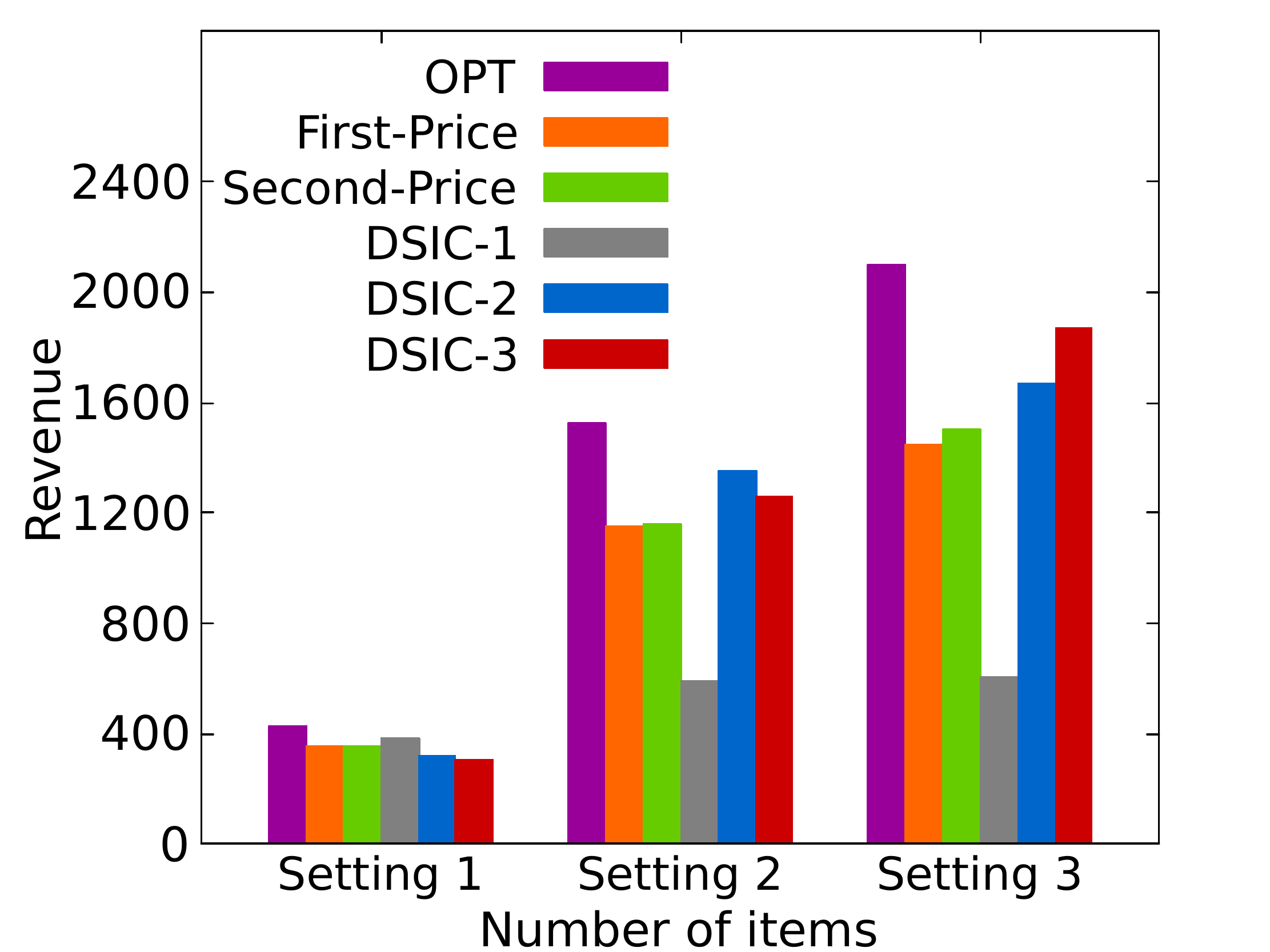}
		\caption{Revenue with 40 mixed bidders in Setting 1: 200 items; Setting 2: 1000 items; Setting 3: 1600 items; with three different groups of rank score functions for our DSIC auction}
		\label{fig:non-iid}
	\end{figure}	
	
	\noindent \textbf{Symmetric bidders }
	The revenue of our truthful auction (referred as DSIC) and the baseline auctions with different number of bidders is reported in Figure \ref{fig:bnum-vary}, and revenue when number of items changes is reported in Figure \ref{fig:inum-vary}. 
	Our truthful auction achieves better revenue than first-price and second-price auctions, and generally achieves more than 90\% revenue of the near optimal baseline. Remarkably, first-price auctions displays evident decrease in revenue when the items are excessive, which also leads to decreases in fairness (Figure \ref{fig:inum-fair}). 
    This is because the bidders exceedingly increase their reported ROI when facing some unsold items with low prices, leading all their payment to decrease proportionally. 
    In contrast, for second-price auctions, the payment of unsold items only consume a negligible part of budget when computing best-response ROI, and the reported ROI of bidders maintains relatively small.
	
	\par In Figure~\ref{fig:inum-fair}, we report the fairness results under different number of items a fixed number of 40 bidders as in Figure \ref{fig:inum-vary}. As designed in the rank score function, our auction does not allocate excessive items to a certain bidder to avoid waste of items, and thus achieves better fairness then the optimal baseline when the number of items is in the range of $600$ to $1200$. 
	When items are in shortage ($200/400$ items in Figure \ref{fig:inum-fair}), our truthful auction harms the fairness to some extent due to the pursuit of revenue. We can control the tradeoff between fairness and revenue with different rank score functions through the parameter $\beta$. Figure \ref{fig:fair-trade} reports a set of revenue and fairness performances with different $\beta$ in the same setting as $400$ items in Figure \ref{fig:inum-fair}. 
	To achieve higher revenue, it is inevitable to have discriminating allocation among different bidders and harms the fairness when items are not abundant. \\
	\noindent \textbf{Mixed bidders }
	Instead of presenting trends similar to i.i.d. settings, we demonstrate the importance of choosing appropriate rank score functions in non-i.i.d. settings. We choose three groups of rank score functions respectively with good performance in three typical settings: 40 mixed bidders with 200/1000/1600 items, and test their performances in the other two settings. The results are presented in Figure \ref{fig:non-iid}, where DSIC-$n$ represents the truthful auction with rank score functions performing well in Setting $n$. The DSIC-$n$ mechanism only perform well in Setting $n$, and the reason should come from the distinct allocation approaches needed to achieve high revenue for different number of items. For example, to achieve high revenue when items are in shortage (Setting 1), DSIC-1 auction allocates most items to the bidders with high value and low ROI, leading to allocating excessive items to these bidders and neglecting other bidders in Setting 2 and 3. 
	
	\section{Related Work}
	\par Auctions in online advertising has attracted much attention from game theory field. Initiated by \cite{varian2007position,edelman2007internet}, the early work studied the properties of keyword auctions in sponsored search, where bidders set a fixed bid for all the queries with a certain keyword. An extensive survey on this topic can be found at \cite{qin2015sponsored}. 	
         \par Due to the unique features compared to traditional ad auctions, auction design and game theoretical analysis for automated bidding come into researchers' view in recent years. Starting from \cite{aggarwal2019autobidding}, a line of work takes the uniform pacing assumption to analyse the equilibrium and efficiency between automated bidding agents, where the agent applies a same pacing parameter to the true value of each item as their bids. \cite{deng2021towards,balseiro2021robust} studied the usage of boosting and reserve price to increase the efficiency and revenue of auto-bidding auctions. \cite{chen2021complexity,conitzer2021multiplicative,conitzer2022pacing} analysed the existence and computational complexity of pacing equilibrium under different settings.  \cite{mehta2022auction} designed a truthful auction with randomization to improve the price of anarchy in auto-bidding settings. \cite{deng2022efficiency} studied the efficiency of first-price auctions under a more generalized setting with no uniform pacing assumption, and \cite{liaw2022efficiency} designed a non-truthful auction to further improve the price of anarchy. Some recent work assumed agents adopting regret-minimizing learning algorithms to adjust the bid or pacing parameter, and analysed the convergence as well as game properties \cite{gaitonde2022budget,fikioris2022liquid,kolumbus2022auctions}.

        \par Our work could be classified as considering the incentive of bidders when setting automated bidding parameters. Under this topic, \cite{li2020incentive} provided distributed bidding algorithms for IC issues assuming public values and private ROI. \cite{balseiro2021landscape} studied the optimal mechanism design for bidders with ROI constraint under different information structure. \cite{kolumbus2021and} considered the incentive of parameter setting for agents in general games. As one of the most relevant work, \cite{DBLP:conf/sigecom/BalseiroDMMZ22} considered the mechanism design with private ROI constraint and public budget, and derived the optimal auction for all the two-bidder cases and some specific multi-bidder cases. Distinct from the above work, our work considers both budget and ROI as private constraints, which aligns with the real automated bidding and reveals the complex value grouping phenomenon for the first time.
        \par Another closely relevant line of work is mechanism design for financially-constrained bidders. Since \cite{laffont1996optimal}, intensive research has been conducted for bidders with budget constraint and progress gradually \cite{borgs2005multi,bhattacharya2010budget,dobzinski2012multi,pai2014optimal}.        \cite{szymanski2006impact} analysed the impact of ROI constraints on bidding and revenue of several common auction forms, and
        \cite{golrezaei2021auction} firstly considered auction design for bidders with ROI constraints. 
        Differing from those existing work studying quasi-linear bidders with private values and one financial constraint, our work consider value-maximizing bidders with multiple private constraints in the automated bidding environment.	
        \par From the computational perspective of multi-dimensional mechanism design, the representing series of work \cite{cai2012optimal,alaei2012bayesian,cai2021efficient} provided a black-box linear programming framework to the classical multi-item auction design problem under Bayesian Incentive Compatibility (BIC). However, due to the ``if, then'' structure in our truthfulness conditions (Theorem 3.5) and the resulting value grouping phenomenon, those existing linear programming techniques could not be simply adapted to our problem.
	
	\section{Conclusion}
	In this work, we have considered truthful auto-bidding auction design for bidders with private budget and ROI constraints across multiple impressions, while the values of impressions are public to the auctioneer. We have characterized the truthfulness conditions in this new auto-bidding auction model, which involves irregular grouping constraints on bidders' cumulative utilities. We have proposed a series of simple truthful auction mechanisms with flexible rank score functions as a solution to this automated bidding auction design problem. The experimental results validate the performances and flexibility of the proposed auction mechanisms.

\bibliographystyle{named}
\bibliography{ijcai23}

\newpage
\appendix

\section{Detailed Proofs for Results}
 During the proofs of results, for ease of notations, we may drop the term of other bidders' reported types $\boldsymbol{t}_{-i}$ and the subscript $i$ in $u$, $v$, and $p$, \emph{e.g.}, $v_i((B_i, R_i), \boldsymbol{t}_{-i})$ to $v(B_i, R_i)$ when the context is clear. With dropped $\boldsymbol{t}_{-i}$ and $i$, we imply that all the mentioned statements should hold for every possible $\boldsymbol{t}_{-i}\in \mathcal{T}_{-i}$ for bidder $i$. We omit the proofs of Theorem 3.1, 3.2 and 3.7 since we have detailedly explained how to derive them in the main text.
 
\subsection{Proof of Theorem 3.3}

\par To help prove Theorem 3.3, we consider the following Lemma \ref{lm:IC-BR-2}, which provides additional insight on the structure of feasible allocations. It describes the following fact: if some type $A$ has both a higher budget and higher ROI than another type $B$ (indicating we cannot simply distinguish who would have higher utility from their types), and they turn out to share the same cumulative value, then the corresponding lower utility diagonal type, \emph{i.e.}, the type who has the lower budget and higher ROI, must have the same cumulative value. 
	
	\begin{lemma} \label{lm:IC-BR-2}
			For any auction satisfying Theorem 3.1 and 3.2, $$\text{if }v_i((B_i', R_i), \boldsymbol{t}_{-i})=v_i((B_i, R_i'), \boldsymbol{t}_{-i})=c$$ for some $c\geq 0, i\in [n]$, $\boldsymbol{t}_{-i}\in \mathcal{T}_{-i}$, $B_i'>B_i$ and $R_i'<R_i$, then $$v_i((B_i, R_i), \boldsymbol{t}_{-i})=c.$$
		
		\end{lemma}
	\begin{proof} 
		By Theorem 3.1 and 3.2, $v_i(B_i, R_i)\leq c$. Suppose $v_i(B_i, R_i)<c$, then $v(B_i',R_i)>v(B_i, R_i)$, indicating $p(B_i', R_i)> B_i$ and $\operatorname{ROI}(B_i', R_i)\geq R_i$. Similarly, $v(B_i, R_i')>v(B_i, R_i)$, indicating $\operatorname{ROI}(B_i, R_i')<R_i$ and $p(B_i, R_i')\leq B_i$. Thus, we have $p(B_i, R_i')\leq B_i < p(B_i', R_i)$ and $\operatorname{ROI}(B_i, R_i') <R_i \leq \operatorname{ROI}(B_i', R_i)$. As $v_i=p_i\times \operatorname{ROI}_i$
		, $v(B_i, R_i')<v(B_i', R_i)$, leading to a contradiction.
	\end{proof}

	\begin{proof} 
		(Theorem 3.3) The only if direction directly follows from Theorem 3.1 and 3.2. For the if direction, we divide the possible misreporting into four categories according to their relationship with the true type $t_i=(B_i. R_i)$, and verify none of them could achieve higher utility when conditions in Theorem 3.1 and 3.2 hold. Suppose the bidder misreports $(B_i', R_i')$. The cases $B_i'=B_i$ or $R_i'=R_i$ are covered directly.\\ 
		\noindent $\bullet$ $B_i'>B_i$ and $R_i'>R_i$. As $R_i'>R_i$, $v(B_i, R_i)\geq v(B_i, R_i')$. If $v(B_i', R_i')>v(B_i, R_i)$, then $v(B_i', R_i')>v(B_i, R_i')$, indicating $p(B_i', R_i')$ $>B_i$. Therefore, if misreporting could bring higher cumulative value, then the payment would break the bidder's budget constraint. \\
		\noindent $\bullet$ $B_i'>B_i$ and $R_i'<R_i$. We would have $v(B_i', R_i')\geq v(B_i, R_i')\geq v(B_i, R_i)$, and $v(B_i', R_i')\geq v(B_i', R_i)\geq v(B_i, R_i)$. If $v(B_i', R_i')> v(B_i, R_i)$, at least one inequality strictly holds respectively in the above two inequality formulas. If $v(B_i', R_i')>v(B_i, R_i')$, then $p(B_i', R_i')>B_i$, which breaks the bidder's budget constraint. Thus, $v(B_i', R_i')= v(B_i, R_i')> v(B_i, R_i)$ should hold. Similarly, $v(B_i', R_i')= v(B_i', R_i)> v(B_i, R_i)$. We then have $v(B_i', R_i')=v(B_i', R_i)=v(B_i, R_i')>v(B_i, R_i)$, which contradicts with Lemma \ref{lm:IC-BR-2}. \\
		\noindent $\bullet$ $B_i'<B_i$ and $R_i'>R_i$. According to condition (1), the bidder gets lower allocated value. \\
		\noindent $\bullet$ $B_i'<B_i$ and $R_i'<R_i$. This case is symmetric to $B_i'>B_i$ and $R_i'>R_i$, and we omit the proof here.
	\end{proof}

\subsection{Proof of Theorem 3.4}
	\begin{proof}
		Since $\mathcal{A}$ is feasible, there exists corresponding $\mathcal{P}'$ that results in DSIC and IR mechanism $(\mathcal{A}, \mathcal{P}')$. Notice $p_i(B_i,R_i)=\min \left( v_i(B_i,R_i)/R_i, B_i \right)$ is the possible largest payment that can be extracted from bidder $i$ with type $(B_i,R_i)$ to satisfy the bidder's ROI and budget requirement. That is to say, $p_i'(B_i,R_i)\leq p_i(B_i,R_i)$. On the other hand, increasing payment of a type will not break originally satisfied conditions in Theorem 3.1 and 3.2. As modifying from $p_i'$ to $p_i(B_i,R_i)=\min \left( v_i(B_i,R_i)/R_i, B_i \right)$ satisfies the budget/ROI requirement of the bidder and maintaining the DSIC conditions, $(\mathcal{A}, \mathcal{P})$ is a truthful mechanism.
		\end{proof}

\subsection{Proof of Theorem 3.5}
	\begin{proof}
		The if direction could be verified through applying Theorem 3.3 and 3.4. For the only if direction, condition (1) directly follows from Theorem 3.3. To show condition (2), if $v(B_i, R_i)< B_i R_i$, then $p(B_i, R_i) \leq v(B_i, R_i)/R_i < B_i$, leading to contradiction with Theorem 3.3 when $B_i'\to B_i^{-}$ (denoting some $B_i'<B_i$ infinitely close to $B_i$). The analysis is similar for condition (3).
	\end{proof}

\subsection{Proof of Theorem 3.6}
	\begin{proof}
		  When $\operatorname{thr}_i(B_i)=0$, the statement is vacuously true. When $\operatorname{thr}_i(B_i)>0$, $$\lim _{R\to \operatorname{thr}_i(B_i)^-}v_i(B_i, R) \geq \operatorname{thr}_i(B_i)\times B_i$$ from definition of $\operatorname{thr}$ function. By the monotone property, for any $R_i<\operatorname{thr}_i(B_i)$, $v_i(B_i, R_i) > B_i\times R_i $, indicating $v_i(B_i, R_i)=\lim _{R\to R_i^+}v_i(B_i, R)$. As a result, $v_i(B_i, R_i)=v_i(B_i, \operatorname{thr}_i(B_i))$ for any $R_i\leq  \operatorname{thr}_i(B_i)$. If $v_i(B_i, \operatorname{thr}_i(B_i))>\operatorname{thr}_i(B_i)\times B_i$, by Theorem 3.5, $$\lim _{R\to \operatorname{thr}_i(B_i)^+} v_i(B_i, R) =v_i(B_i, \operatorname{thr}_i(B_i)) >\operatorname{thr}_i(B_i)\times B_i.$$ If $v_i(B_i, \operatorname{thr}_i(B_i)) < \operatorname{thr}_i(B_i)\times B_i$, there would exist some $R\to \operatorname{thr}_i(B_i)^-$ whose $v_i(B_i, R)<B_i \times R$. Both cases contradicts of the definition of supremum in $\operatorname{thr}$ function. Thus, $v_i(B_i, \operatorname{thr}_i(B_i)) = \operatorname{thr}_i(B_i)\times B_i$.
	\end{proof}

\subsection{Proof of Corollary 3.8}
	\begin{proof}
		We consider the map $m$ as follows.
		\begin{equation*}
			(m \circ g) (B, R)=\left\{
			\begin{aligned}
				&g(B) \text{\quad \quad \quad \quad \quad \quad \quad \quad if } R\leq \frac{g(B)}{B}& \\
				& g\left(\sup_{B'\leq B} \{R\leq \frac{g(B')}{B'}\} \right) \text{\quad otherwise,}
			\end{aligned}
			\right.
		\end{equation*}
		where we define $\sup \varnothing = 0$. Given a non-decreasing function $g$, we could derive the corresponding unique cumulative value function $v$ based on the construction of $m$, and the constructed cumulative value function must be feasible since every type satisfies conditions in Theorem 3.5. Inversely, given any possible cumulative value function $v_i$, we can find the unique threshold ROI function by definition
  $$\operatorname{thr}_i(B, \boldsymbol{t}_{-i})=\sup _{R \in \mathcal{R}_i} \left( v_i\left((B, R), \boldsymbol{t_{-i}}\right)\geq B\times R \right),$$ and compute the corresponding  $$g_i(B_i, \boldsymbol{t}_{-i}) =\operatorname{thr}_i(B_i, \boldsymbol{t}_{-i})\times B_i.$$ Due to the the truthfulness conditions in Theorem 3.5, such derived $g_i(B_i, \boldsymbol{t}_{-i})$ must be non-decreasing in $B_i$ since the cumulative value when $R$ is infinitely close to $0$ should be non-decreasing for different $B_i$.
	\end{proof}

 \subsection{Proof of Theorem 4.1}
 \begin{proof}
 We provide the proof through figuring out the function $g(B)$ for each bidder $i$ when $\boldsymbol{t}_{-i}$ is fixed. Fixing bidder $i$, given $c_j$ and $v_{i,j}$, we rank the items by $r_{i,j}=f_{i,j}^{-1}(\frac{c_j}{v_{i,j}})$ in an decreasing order. Without loss of generality, we renumber these items from $1$ to $m$ with $r_i^1 > r_i^2 > \ldots > r_i^{m}$ and denote the corresponding value as $v_i^j$. 
\par We would then verify that the allocation assigned by our algorithm aligns with the cumulative values implied by following function $g(B)$ with $p = 1,2,\ldots, m-1$. The corresponding threshold ROI function is $\operatorname{thr}(B)=g(B)/B$.
\begin{equation*}
g(B)=\left\{
    \begin{aligned}
    & r_{i}^1 \cdot B \text{\quad \quad \quad \quad \quad \quad \quad \quad \quad \quad if } B< \frac{v_i^1}{r_{i}^1},& \\
    & \sum _{j=1} ^{p} v_{i}^j \text{\quad \quad \quad if } B \in [\frac{\sum _{j=1}^p v_{i}^j}{r_{i}^p}, \frac{\sum _{j=1}^p v_{i}^j}{r_{i}^{p+1}}), & \\
    & r_{i}^{p+1} \cdot B \text{\quad \quad if } B \in [\frac{\sum _{j=1}^p v_{i}^j}{r_{i}^{p+1}}, \frac{\sum _{j=1}^{p+1} v_{i}^j}{r_{i}^{p+1}}), & \\
	& \sum _{j=1} ^{m} v_i^{j} \text{\quad \quad \quad \quad \quad \quad \quad if } B \geq \frac{\sum _{j=1} ^{m} v_i^{j} }{r_{i}^{m}}.&
	\end{aligned}
	\right.
\end{equation*}
        \par For convenience in notations, we use $A_i(R)$ to denote the winning set $A_i$ containing all items the bidder has the highest virtual bid when she reports $R$ (Line 4-6), namely all the items with $r_i^j\geq R$, and use $A_i(0)$ to denote the whole set of items. We use $R_i^c(B;A_i(R))$ to denote the critical ROI $R_i^c$ computed after Line 9 when bidder $i$ reports $(B,R)$, and use $d(B;A_i(R))$ to represent the corresponding $d$.
	\par We would show the following properties to finish the proof. 
 \begin{enumerate}[(a)]
     \item $R_i^c(B;A_i(0))$ aligns with the threshold ROI function indicated by $g(B)$;
     \item When $R\leq R_i^c(B;A_i(0))$, $$ R_i^c(B;A_i(0))=R_i^c(B;A_i(R)),$$ $$d(B;A_i(0))=d(B;A_i(R)),$$ and when $R>  R_i^c(B;A_i(0))$, $R>R_i^c(B;A_i(R))$. 
     \item When $R\leq R_i^c(B;A_i(R))$, $R\leq g(B)$ and the cumulative values assigned to the bidder align with those defined by $g(B)$;
     \item When $R>R_i^c(B;A_i(R))$, $R> g(B)$ and the bidder gets all the items in $A_i(R)$;
     \item The corresponding allocation indicated by (d) align with those indicated by $g(B)$.
 \end{enumerate}
 
        \par To show property (a), we analyse the detailed calculation in Line 9.
 The calculation in Line 9 exactly locates budget $B$ in the segmentation of above function $g(B)$. The left-hand side of Line 9 is strictly increasing when $R_i^c$ decreases, and is continuous almost everywhere except at points $R_i^c=\{r_{i}^j\}_{j\in [m]}$. In detail, when $R_i^c$ decreases from $(r_i^{p+1})^+$ to $(r_i^{p+1})$, the left-hand side of Line 9 jumps from $B=\left(\frac{\sum _{j=1}^p v_{i}^j}{r_{i}^{p+1}}\right)^-$ to $\frac{\sum _{j=1}^{p+1} v_{i}^j}{r_{i}^{p+1}}$, where $(\cdot)^+$ or $(\cdot)^-$ represents the one-sided limits. Therefore, we have
{  \small
\begin{equation*}
R_i^c(B;A_i(0))=\left\{
    \begin{aligned}
    & r_{i}^1  \text{\quad \quad \quad \quad \quad \quad \quad \quad \quad \quad \quad \quad \quad if } B< \frac{v_i^1}{r_{i}^1},& \\
    & \left(\sum _{j=1} ^{p} v_{i}^j\right)/B \text{\quad if } B \in [\frac{\sum _{j=1}^p v_{i}^j}{r_{i}^p}, \frac{\sum _{j=1}^p v_{i}^j}{r_{i}^{p+1}}), & \\
    & r_{i}^{p+1} \text{\quad \quad \quad \quad \quad if } B \in [\frac{\sum _{j=1}^p v_{i}^j}{r_{i}^{p+1}}, \frac{\sum _{j=1}^{p+1} v_{i}^j}{r_{i}^{p+1}}), & \\
	& \left(\sum _{j=1} ^{m} v_i^{j}\right)/B \text{\quad \quad \quad \quad \quad \quad if } B \geq \frac{\sum _{j=1} ^{m} v_i^{j} }{r_{i}^{m}},&
	\end{aligned}
	\right.
\end{equation*}  
}
and
{ \small
\begin{equation*}
d(B;A_i(0))=\left\{
    \begin{aligned}
    & v_i^1/r_{i}^1 - B  \text{\quad \quad \quad \quad \quad \quad \quad \quad \quad \quad if } B< \frac{v_i^1}{r_{i}^1},& \\
    & 0 \text{\quad \quad \quad \quad \quad \quad \quad if } B \in [\frac{\sum _{j=1}^p v_{i}^j}{r_{i}^p}, \frac{\sum _{j=1}^p v_{i}^j}{r_{i}^{p+1}}), & \\
    & \sum _{j=1}^{p+1} v_{i}^j/r_{i}^{p+1} - B \text{\quad if } B \in [\frac{\sum _{j=1}^p v_{i}^j}{r_{i}^{p+1}}, \frac{\sum _{j=1}^{p+1} v_{i}^j}{r_{i}^{p+1}}), & \\
	& 0 \text{\quad \quad \quad \quad \quad \quad \quad \quad \quad \quad \quad if } B \geq \frac{\sum _{j=1} ^{m} v_i^{j} }{r_{i}^{m}}.&
	\end{aligned}
	\right.
\end{equation*} 
}
Comparing $\operatorname{thr}(B)=g(B)/B$ with $R_i^c(B;A_i(0))$ gives us $R_i^c(B;A_i(0))=\operatorname{thr}(B)$, thus proving property (a). 
\par To show property (b), we should notice the calculation of $R_i^c(B; A(R))$ and $d(B; A(R))$ relies on the items with $r_i^j\geq R_i^c(B; A(R))$, and the existence of other items in $A(R)$ would not affect the result according to the way we compute $R_i^c$. This provides the proof for the first half of property (b). To show the second half, note that a larger set $A_i$ could only result in an equal or larger $R_i^c$, and $A_i(R)\subset A_i(0)$.

\par When $R\leq R_i^c(B;A_i(R))$, $R\leq g(B)$ follows from property (b). The algorithm goes into Line 13, thus items with $r_i^j<R_i^c(B;A_i(R))$ would be removed. After the operations to remove items with value $d/R_i^c$ in Line 10-11, the values of remaining items would finally become $\sum _{j, r_i^j\geq R_i^c(B)} v_i^j -d(B)/R_i^c(B)$, where all the set used to compute $R_i^c$ and $d$ here could be regarded as $A_i(0)$ by property (b). It can be verified that the above values align with our defined $g(B)$, thus proving property (c). 

\par When $R> R_i^c(B;A_i(R))$, $R> g(B)$ follows from property (b), and the algorithm would not go into Line 13 to remove parts of the items. This also means $B \geq \frac{\sum _{j=1} ^{n_i} v_i^{j} }{r_{i}^{n_i}}$, where we use $n_i$ to denote the item with smallest $r_i^j$ in $A_i(R)$, and further indicates $d(B;A_i(R))=0$. Therefore, the algorithm would not go into Line 10-11, and no operation about removing the items would be conducted. The bidder would get all items in $A_i(R)$ after Line 6, which proves property (d).

\par By property (d), we naturally find that the cumulative values are the same for types with the same ROI and locate above the threshold ROI function. This structure aligns with the non-increasing threshold ROI function indicated by the defined $g(B)$. Based on Corollary 3.8, in order to prove property (e), we only need to show the type which has the largest budget among all the types with ROI $R$ on the threshold ROI function, \emph{i.e.}, the type $(B^*(R), R)$ with 
$B^*(R)=\sup_{B'} \{B'|\operatorname{thr}(B')=R\}$, would get an cumulative value equal to the valuations of items in $A_i(R)$ (as characterized in property (d)). Since we have shown the cumulative values for types below and on the threshold ROI function are aligned between our algorithm and defined $g(B)$ in properties (a)-(c), we could use the defined $g(B)$ and corresponding $\operatorname{thr}(B)$ to verify the above statement. We could find
\begin{equation*}
B^*(R)=\left\{
    \begin{aligned}
    & \frac{\sum _{j=1}^{p} v_{i}^j}{r_{i}^{p}}  \text{\quad \quad \quad \quad \quad \quad \quad \quad \quad \quad if } R=r_i^p, & \\
    & \sum _{j=1}^{p} v_{i}^j/R \text{\quad \quad \quad \quad \quad \quad \quad if } R \in (r_i^{p+1}, r_i^p), & \\
    & \sum _{j=1}^{m} v_{i}^j/R \text{\quad \quad \quad \quad \quad \quad \quad \quad \quad \quad if } R\leq r_i^{m}, & \\    
	\end{aligned}
	\right.
\end{equation*} 
for $p=1, \ldots, m-1$. We do not define $B^*(R)$ for $R> r_i^1$ since both our algorithm and the defined $g(B)$ determine to not allocate to these types. We could verify $g\left(B^*(R)\right)=\sum _{j, r_i^j\geq R} v_i^j$, thus finish the proof for property (e). 

	\par By the above properties (a)-(e), the cumulative value function indicated by our algorithm aligns with the defined $g(B)$. Since this $g(B)$ is non-decreasing in $B$, the mechanism is truthful and preserves the bidders' financial constraints with the payment in Line 14. Finally, our mechanism does not over-allocate items because each item may only be allocated to the bidder with highest virtual bid (breaking possible ties with arbitrary rules).
\end{proof}

\section{Non-Truthfulness of Repeated First-Price and Second-Price Auctions}
We provide counter-examples to show the non-truthfulness of repeated first-price and second-price auctions. For simplicity of illustration and understanding, we consider small cases with several bidders and items that do not fit with real situations in automated bidding, while the ideas behind the non-truthfulness are the same.
\subsection{Repeated First-Price Auctions}
In repeated first-price auction, the realized ROI of a bidder is exactly her reported ROI. Thus, bidders have no incentive to misreport a smaller ROI. However, reporting a larger ROI can possibly bring the bidder higher utility since it could reduce the payment for per unit gained value if the bidder could still win enough items.
\begin{example}
Consider a simple case with two bidders and two items shown in the following table.
    	\begin{table}[H]
    	\centering
        	\begin{tabular}{l l l l l}
            \toprule
             & Item 1 & Item 2 & ROI & Budget  \\
            \midrule
            Bidder 1  & $v_{1,1}=4$ & $v_{1,2}=4$   & $R_1=2$  & $B_1=3$   \\
            Bidder 2  & $v_{2,1}=1$ & $v_{2,2}=1$   & $R_2=1.5$  & $B_2=6$ \\
            \bottomrule
        	\end{tabular}
    	\end{table}
\noindent Assume bidder 2 reports her real type, and we consider the strategy of bidder 1. If bidder 1 reports her real type, she will only win item 1 at price of 2, leading the remaining budget to be 1, thus cannot afford the second item. However, if bidder 1 reports her ROI higher as $R_1=4$, she could win both items at total price 2 and realized ROI 4, getting a higher utility. 
\end{example}
\subsection{Repeated Second-Price Auctions}
In repeated second-price auctions, both misreporting a larger or smaller ROI might bring the bidder larger utility.  

\begin{example}
Consider a simple case with two bidders and two items shown in the following table.
    	\begin{table}[H]
    	\centering
        	\begin{tabular}{l l l l l l}
            \toprule
             & Item 1 & Item 2 & ROI & Budget  \\
            \midrule
            Bidder 1  & $v_{1,1}=4$ & $v_{1,2}=8$  & $R_1=1$  & $B_1=3$   \\
            Bidder 2  & $v_{2,1}=4$ & $v_{2,2}=4$  & $R_2=1.5$  & $B_2=6$ \\
            \bottomrule
        	\end{tabular}
    	\end{table}
\noindent Assume bidder 2 reports her real type, and we consider the strategy of bidder 1. If bidder 1 reports her real type, she will win item 1 at price $8/3$, leading the remaining budget to be $1/3$, thus cannot afford the remaining higher value item 2. However, if bidder 1 reports her ROI higher as $R_1=2$, then she will lose item 1, and win item 2 at price $8/3$ with realized ROI 3, getting a higher utility.
\end{example}
Reporting a smaller ROI may help the bidder win more items and balance the realized ROI between items (items with realized ROI higher or lower than the real ROI both exist) without breaking her true ROI.
\begin{example}
Consider a simple case with two bidders and two items shown in the following table.
    	\begin{table}[H]
    	\centering
        	\begin{tabular}{l l l l l}
            \toprule
             & Item 1 & Item 2 & ROI & Budget  \\
            \midrule
            Bidder 1  & $v_{1,1}=4$ & $v_{1,2}=3$   & $R_1=2$  & $B_1=6$   \\
            Bidder 2  & $v_{2,1}=1$ & $v_{2,2}=4$   & $R_2=2$  & $B_2=6$ \\
            \bottomrule
        	\end{tabular}
    	\end{table}
\noindent Assume bidder 2 reports her real type, and we consider the strategy of bidder 1. If bidder 1 reports her real type, she will win item 1 at price of 0.5 and lose item 2 since bidder 2 has a higher $v/R$ for it. However, if bidder 1 reports her ROI lower as $R_1=1.2$, then she could win both items at total price $2.5$ and realized ROI 2.8, getting a higher utility.
\end{example}

\end{document}